\documentclass{IEEEtran}

\usepackage[english]{babel}
\usepackage[utf8x]{inputenc}
\usepackage[T1]{fontenc}

\usepackage{amsmath,amsthm}
\usepackage{amssymb}
\usepackage{multirow}
\usepackage{ccaption} 

\usepackage{amsmath}
\newtheorem{theorem}{Theorem}[section]

\newtheorem{lemma}[theorem]{Lemma}
\usepackage{graphicx}
\usepackage[colorinlistoftodos]{todonotes}
\theoremstyle{definition}
\theoremstyle{plain}

\title{An Approximate Pareto Set for Minimizing the Maximum Lateness and Makespan on Parallel Machines}

\author{Gais Alhadi$^{1}$, Imed Kacem$^{2}$, Pierre Laroche$^{3}$, and Izzeldin M. Osman$^{4}$% <-this % stops a space
\thanks{$^{1}$Gais Alhadi is a member of Faculty of Mathematical and Computer Sciences, University of Gezira, Wad-Madani, Sudan
        {\tt\small gais.alhadi@uofg.edu.sd}}%        
\thanks{Imed Kacem$^{2}$ and Pierre Laroche$^{3}$ are members of the LCOMS Laboratory, University of Lorraine, F-57045 Metz, France 
        {\tt\small imed.kacem@univ-lorraine.fr}, {\tt\small pierre.laroche@univ-lorraine.fr}}%
\thanks{Izzeldin M. Osman$^{4}$ is from Sudan University of Science and Technology, Khartoum, Sudan 
		{\tt\small izzeldin@acm.org}}%        
}

\begin{document}
\maketitle

%% ***************************************************** abstract **********************************************

\begin{abstract}
We consider the two-parallel machines scheduling problem, with the aim of minimizing the maximum lateness and the makespan. Formally, the problem is defined as follows. We have to schedule a set $J$ of $n$ jobs on two identical machines. Each job $i\in J$ has a processing time $p_i$ and a delivery time $q_i$. Each machine can only perform one job at a given time. The machines are available at time $t=0$ and each of them can process at most one job at a given time. The problem is to find a sequence of jobs, with the objective of minimizing the maximum lateness $L_{max}$ and the makespan $C_{max}$. With no loss of generality, we consider that all data are integers and that jobs are indexed in non-increasing order of their delivery times: $q_1\geq q_2\geq\ldots\geq q_n$. This paper proposes an exact algorithm (based on a dynamic programming) to generate the complete Pareto Frontier in a pseudo-polynomial time. Then, we present an FPTAS (Fully Polynomial Time Approximation Scheme) to generate an approximate Pareto Frontier, based on the conversion of the dynamic programming. The proposed FPTAS is strongly polynomial. Some numerical experiments are provided in order to compare the two proposed approaches.

\end{abstract}

%\textit{\textbf{keyword}:}
%\textit{Approximation, Makespan, Maximum lateness, Dynamic programming, FPTAS. }

%% ***************************************************** Introduction **********************************************

\section{Introduction}

We consider the two-parallel machines scheduling problem, with the aim of minimizing the maximum lateness and makespan. Formally, the problem is defined as follows. We have to schedule a set $J$ of $n$ jobs on two identical machines. Each job $i\in J$ has a processing time $p_i$ and a delivery time $q_i$. The machines are available at time t=0 and each of them can process at most one job at a time. The problem is to find a sequence of jobs, with the objective of minimizing the maximum lateness $L_{max}$ and the makespan $C_{max}$. With no loss of generality, we consider that all data are integers and that jobs are indexed in non-increasing order of their delivery times $q_1\geq q_2\geq\ldots\geq q_n$.\\
For self-consistency, we recall some necessary definitions related to the approximation area. An algorithm $A$ is called a $\rho-$\textit{approximation algorithm} for a given problem, if for any instance $I$ of that problem the algorithm $A$ yields, within a polynomial time, a feasible solution with an objective value $A(I)$ such that: $|A(I)-  OPT(I)|\leq \epsilon.OPT(I)$, where $OPT(I)$ is the optimal value of $I$ and $\rho$ is the performance guarantee or the worst-case ratio of the \textit{approximation algorithm} $A$. It can be a real number greater or equal to $1$ for the minimization problems $\rho =1+ \epsilon$ (that it leads to inequality $A(I)\leq (1+ \epsilon)OPT(I)$), or it can be real number from the interval $[0,1]$ for the maximization problems $\rho =1-\epsilon$ (that it leads to inequality $A(I)\geq(1- \epsilon)OPT (I)$). The Pareto-optimal solutions are the solutions that are not dominated by other solutions. Thus, we can consider that the solution is Pareto-optimal if there does not exist another solution that is simultaneously the best for all the objectives. Noteworthy, Pareto-optimal solutions represent a range of reasonable optimal solutions for all possible functions based on the different objectives. A schedule is called Pareto-optimal if it is not possible to decrease the value of one objective without increasing the value of the other.\\
It is noteworthy that during the last decade the multi-objective scheduling problems have attracted numerous researchers from all the world and have been widely studied in the literature. 
For the scheduling problems on serial-batch machine, Geng et al.\cite{Geng_Yuan} studied scheduling problems with or without precedence relations, where the objective is to minimize makespan and maximum cost. They have provided highly efficient polynomial-time algorithms to generate all Pareto optimal points. An approximate Pareto set of minimal size that approximates within an accuracy $\epsilon$ for multi-objective optimization problems have been studied by Bazgan et al.\cite{Bazgan_et_al}. They proposed a 3-approximation algorithm for two objectives and also proposed a study of the greedy algorithm performance for a three-objective case when the points are given explicitly in the input. They showed that the three-objective case is NP-hard. Chen and Zou \cite{Chen_Zou} proposed a runtime analysis of a $(\mu+1)$ multi-objective evolutionary algorithm for three multi-objective optimization problems with unknown attributes. They showed that when the size of the population is less than the total number of Pareto-vector, the $(\mu+1)$ multi-objective evolutionary algorithm cannot obtain the expected polynomial runtime for the exact discrete multi-objective optimization problems. Thus, we must determine the size of the population equal to the total number of leading ones, trailing zeros. Furthermore, the expected polynomial runtime for the exponential discrete multi-objective optimization problem can be obtained by the ratio of $n/2$ to $\mu-1$ over an appropriate period of time. They also showed that the $(\mu+1)$ multi-objective evolutionary algorithm can be solved efficiently in polynomial runtime by obtaining an $\epsilon-$ adaptive Pareto front. Florios and Mavrotas \cite{Florios_Mavrotas} used AUGMECON2, a multi-objective mathematical programming method (which is suitable for general multi-objective integer programming problems), to produce all the Pareto-optimal solutions for multi-objective traveling salesman and set covering problems. They showed that the performance of the algorithm is slightly better than it already exists. Moreover, they showed that their results can be helpful for other multi-objective mathematical programming methods or even multi-objective meta-heuristics. In \cite{Sabouni_Jolai}, Sabouni and Jolai proposed an optimal method for the problem of scheduling jobs on a single batch processing machine to minimize the makespan and the maximum lateness. They showed that the proposed method is optimal when the set with maximum lateness objective has the same processing times. They also proposed an optimal method for the group that has the maximum lateness objective and the same processing times. Geng et al.\cite{Geng_Yuan_Pareto} considered the scheduling problem on an unbounded p-batch machine with family jobs to find all Pareto-optimal points for minimizing makespan and maximum lateness. They presented a dynamic programming algorithm to solve the studied problem. He et al.\cite{He_Lin_lin_Bounded} showed that the Pareto optimization scheduling problem on a single bounded serial-batching machine to minimize makespan and maximum lateness is solvable in $O(n^6)$. They also presented an $O(n^3)$- time algorithm to find all Pareto optimal solutions where the processing times and deadlines are agreeable. 
For the bi-criteria scheduling problem, He et al.\cite{He_Lin_Tian} showed that the problem of minimizing maximum cost and makespan is solvable in $O(n^5)$ time. The authors presented a polynomial-time algorithm in order to find all Pareto optimal solutions. Also, He et al.\cite{He_Wand_Lin} showed that the bi-criteria batching problem of minimizing maximum cost and makespan is solvable in $O(n^3)$ time. The bi-criteria scheduling problem on a parallel-batching machine to minimize maximum lateness and makespan have been considered in \cite{He_et_al}. The authors presented a polynomial-time algorithm in order to find all Pareto optimal solutions. Allahverdi and Aldowaisan \cite{Allahverdi_Aldowaisan} studied the no-wait flow-shop scheduling problem with bi-criteria of makespan or maximum lateness. They also proposed a dominance relation and a branch-and-bound algorithm and showed that these algorithms are quite efficient.

The remainder of this paper is organized as follows. In Section 2, we describe the proposed dynamic programming (DP) algorithm. Section 3, provides the description and the analysis of the FPTAS. In Section 4, we present a practical example for DP and FPTAS. Finally, Section 5 concludes the paper.
%% ********************************************* Dynamic Programming Algorithm ********************************************
\section{Dynamic Programming Algorithm}
\label{sec:DP}
The following dynamic programming algorithm $A$, can be applied to solve exactly this problem. This algorithm $A$ generates iteratively some sets of states. At every iteration $i$, a set $\chi_i$ composed of states is generated $(0 \leq i\leq n)$. Each state $[k,L_{max},C_{max}]$ in $\chi_i$ can be associated to a feasible partial schedule for the first $i$ jobs. Let variable $k\in \{0,1\}$ denote the most loaded machine, $L_{max}$ denote the maximum lateness and $C_{max}$ denote the maximum completion time of the corresponding schedule. The dynamic programming algorithm can be described as follows.

\section*{Algorithm $A$}
\begin{enumerate}
\item 	Set $\chi_1= \{[1,p_1+q_1,p_1]\}$. 
\item 	For $i\in \{2,3,...,n\}$,
\begin{enumerate}
\item $\chi_i= \phi$.
\item For every state $[k,L_{max},C_{max}]$ in $\chi_{i-1}:$
	\begin{itemize}
	\item (schedule job $i$ on machine $k$)
    \item []add $[k$,$max\{L_{max}$,$C_{max} + p_i+q_i\}$,$C_{max} +p_i]$ to $\chi_i$
    \item (schedule job $i$ on machine $1 - k$)
    \item[]if $(C_{max} \geq \sum_{j=1}^{i}p_j - C_{max} )$
		\begin{itemize}
		\item add $[k$,$max\{L_{max}$,$\sum_{j=1}^{i}p_j - C_{max} +q_i\}$,$C_{max}]$ to $\chi_i$
		\\else
		\item add $[1-k$,$max\{L_{max}$,$\sum_{j=1}^{i}p_j - C_{max} +q_i\}$,$\sum_{j=1}^{i}p_j - C_{max}]$ to $\chi_i$ 
		\end{itemize}
	\end{itemize}
\item For every $k$, for every $C_{max}$: keep only one state with the smallest possible $L_{max}.$
\item Remove $\chi_{i-1}$.
\end{enumerate}
\item Return the Pareto front of $\chi_n$, by only keeping non-dominated states.
\end{enumerate}
\textbf{\textit{Remark}}: To destroy the symmetry, we start by $\chi_1= \{[1,p_1+q_1,p_1]\}$ (i.e., we perform job 1 on the first machine). 

%% ****************************************************Approximate Pareto Frontier*******************************************
\section{Approximate Pareto Frontier}
\label{sec:FPTAS}
The main idea of the Approximate Pareto Frontier is to remove a special part of the states generated by the dynamic programming algorithm $A$. Therefore, the modified algorithm ${A}'$ described in Lemma \ref{LemmaFPTAS} produces an approximation solution instead of the optimal solution.\\
Given an arbitrary $\epsilon >0$, we define the following parameters:
		\[\delta_1 = \frac{\epsilon P/2}{n},\]
        and
        \[\delta_2 = \frac{\epsilon(P+q_{max})/3}{n}.\]
where $q_{max}$ is the maximum delivery time and $P$ is the total sum of processing times. \\\\
Let $L_{max}^*$ and $C_{max}^*$ be the optimal solutions for our two objectives. 
%For the same machine $k \in \{0,1\}$ l
Let $LMAX$ and $CMAX$ be the upper bounds for the two considered criteria (scheduling all the jobs on the same machine), such that,
		\[0\leq LMAX=P+q_{max} \leq 3L_{max}^* \] 
        \[0\leq CMAX=P \leq 2C_{max}^* \]
We divide the intervals $[0,CMAX]$ and $[0,LMAX]$ into equal sub-intervals respectively of lengths $\delta_1$ and $\delta_2$. Then, an FPTAS is defined by following the same procedure as in the dynamic programming, except the fact that it will keep only one representative state for every couple of the defined subintervals produced from $[0,CMAX]$ and $[0,LMAX]$. Thus, our
FPTAS will generate approximate sets $\chi_i^\#$of states instead of $\chi_i$. The following lemma shows the closeness of the result generated by the FPTAS compared to the dynamic programming.

\begin{lemma}
\label{LemmaFPTAS}
\textit{For every state $[k,L_{max},C_{max}]\in \chi_i$ there exists at least one approximate state  $[m,L_{max}^\#,C_{max}^\#] \in \chi_i^\#$ such that:}
		\[L_{max}^\#≤L_{max}+i.\max\{\delta_1,\delta_2\}, \]
        and
			\[C_{max}-i.\delta_1 \leq C_{max}^\#≤C_{max}+i.\delta_1.\]
\end{lemma}
\begin{proof}
By induction on $i$.\\
First, for $i = 0$ we have $\chi_i^\# = \chi_1$. Therefore, the statement is trivial.  Now, assume that the lemma holds true up to level $i-1$. Consider an arbitrary state $[k,L_{max},C_{max}] \in \chi_i$.
Algorithm $A$ introduces this state into $\chi_i$ when job $i$ is added to some feasible state for the first $i-1$ jobs. Let $[k^{'},L_{max}^{'},C_{max}^{'}]$ be the above feasible state. Three cases can be distinguished:

\begin{enumerate}
	\item  $[k,L_{max},C_{max}]= [k^{'}$, $\max\{L_{max}^{'}$,$C_{max}^{'} + p_i+q_i\}$, $C_{max}^{'} +p_i]$
	\item  $[k,L_{max},C_{max}]= [k^{'}$, $\max\{L_{max}^{'}$,$\sum_{j=1}^{i}p_j - C_{max}^{'} +q_i\}$, $C_{max}^{'}]$ 
	\item  $[k,L_{max},C_{max}]= [1-k^{'}$, $\max\{L_{max}^{'}$,$\sum_{j=1}^{i}p_j - C_{max}^{'} +q_i\}$,$\sum_{j=1}^{i}p_j-C_{max}^{'}]$ 
\end{enumerate}
We will prove the statement for level $i$ in the three cases.

%% ****************************************************Case 1*******************************************

\begin{itemize}
\item \textbf{$1^{st}$ Case: } $[k,L_{max},C_{max}]=[k^{'}$, $\max\{L_{max}^{'}$,$C_{max}^{'} + p_i+q_i\}$, $C_{max}^{'} +p_i]$ \end{itemize}
Since $[k^{'}$,$L_{max}^{'}$,$C_{max}^{'}]\in \chi_{i-1}$, there exists 
$[k^{'\#},L_{max}^{'\#},C_{max}^{'\#} ]\in \chi_{i-1}^\#$, such that:\\
\(L_{max}^{'\#} \leq L_{max}^{'}+(i-1)\max\{\delta_1,\delta_2\}\)
and \( C_{max}^{'}- (i-1) \delta_1 \leq C_{max}^{'\#} \leq C_{max}^{'} +(i-1) \delta_1 .\) \\
Consequently, the state $[k^{'\#},\max\{L_{max}^{'\#},C_{max}^{'\#}+p_i+q_i\},C_{max}^{'\#}+p_i]$ is created by algorithm $A^{'}$at iteration $i$. However, it may be removed when reducing the state subset. Let $[\alpha,\lambda,\mu]$ be the state in $\chi_i^\#$ that is in the same box as the sate $[k^{'\#},\max\{L_{max}^{'\#},C_{max}^{'\#}+p_i+q_i\},C_{max}^{'\#}+p_i]$. Hence, we have:

\begin{eqnarray}
\lambda & \leq & \max\{L_{max}^{'\#},C_{max}^{'\#}+p_i+q_i \} + \delta_2 \nonumber \\
& \leq & \max\{L_{max}^{'},C_{max}^{'}+p_i+q_i\}\nonumber \\
& &+(i-1).\max\{\delta_1,\delta_2\} + \delta_2 \nonumber \\
& \leq & L_{max} + i.\max\{\delta_1,\delta_2\} \label{eq:case11} 
\end{eqnarray}In addition,
\begin{eqnarray}
\mu & \leq & C_{max}^{'\#}+p_i + \delta_1 \nonumber \\
& \leq & C_{max}^{'}+(i-1)\delta_1 + p_i + \delta_1 = C_{max} + i\delta_1. \nonumber \\  \label{eq:case12}
\end{eqnarray}
and,
\begin{eqnarray}
\mu  &\geq & C_{max}^{'\#}+p_i - \delta_1 \geq C_{max}^{'}-(i-1)\delta_1 + p_i - \delta_1  \nonumber \\
&\geq & C_{max} - i\delta_1.  \nonumber \\
\label{eq:case13}
\end{eqnarray}Consequently, $[\alpha,\lambda,\mu]$ is an approximate state verifying the two conditions.

%% ****************************************************Case 2*******************************************

\begin{itemize}
\item \textbf{$2^{nd}$ Case: } $[k,L_{max},C_{max}]= [k^{'}$,$max\{L_{max}^{'}$,$\sum_{j=1}^{i}p_j - C_{max}^{'} +q_i\}$,$C_{max}^{'}]$ 
\end{itemize}
Since $[k^{'}$,$L_{max}^{'}$,$C_{max}^{'}]\in \chi_{i-1}$, there exists
$[k^{'\#},L_{max}^{'\#},C_{max}^{'\#} ]\in \chi_{i-1}^\#$, such that:\\
\(L_{max}^{'\#} \leq L_{max}^{'}+(i-1)max\{\delta_1,\delta_2\}\)
and \(C_{max}^{'} -(i-1) \delta_1 \leq C_{max}^{'\#} \leq C_{max}^{'} +(i-1) \delta_1 .\)\\\\
Consequently, two sub-cases can occur:
%% **************************************************** Subcase 2.1*******************************************
\begin{itemize}
\item Sub-case 2.1:$\sum_{j=1}^{i}p_j - C_{max}^{'\#} \leq  C_{max}^{'\#}$
\end{itemize}
Here, the state 
 $[k^{'\#}$, $\max\{L_{max}^{'\#}$,$\sum_{j=1}^{i}p_j - C_{max}^{'\#} +q_i\}$,$C_{max}^{'\#}]$ is created by algorithm $A^{'}$at iteration $i$. However, it may be removed when reducing the state subset. Let $[\alpha,\lambda,\mu]$ be the state in $\chi_i^\#$ that is in the same box as $[k^{'\#}$, $\max\{L_{max}^{'\#}$,$\sum_{j=1}^{i}p_j - C_{max}^{'\#} +q_i\}$, $C_{max}^{'\#}]$. Hence, we have:

\begin{eqnarray}
\lambda & \leq & \max\{ L_{max}^{'\#},\sum_{j=1}^{i}p_j - C_{max}^{'\#} +q_i\} + \delta_2 \nonumber \\
& \leq & \max\{L_{max}^{'}+(i-1)\max\{\delta_1,\delta_2\},\sum_{j=1}^{i}p_j \nonumber \\
&& - (C_{max}^{'} -(i-1)\delta_1)+q_i\} + \delta_2 \nonumber \\
& \leq & \max\{L_{max}^{'},\sum_{j=1}^{i}p_j - C_{max}^{'}+q_i\} \nonumber \\
&& +(i-1)\max\{\delta_1,\delta_2\} + \delta_2 \nonumber \\
& \leq & L_{max} +(i-1)\max\{\delta_1,\delta_2\} + \delta_2 \nonumber \\
& < & L_{max}+ i.\max\{\delta_1,\delta_2\} \label{eq:case211} 
\end{eqnarray}

Moreover, 
\begin{eqnarray}
\mu \leq C_{max}^{'\#} + \delta_1 \leq C_{max}^{'}+(i-1)\delta_1 + \delta_1 = C_{max}+i\delta_1 .  \label{eq:case212}
\end{eqnarray}
And,
\begin{eqnarray}
\mu \geq C_{max}^{'\#} - \delta_1 \geq C_{max}^{'} - (i-1)\delta_1 - \delta_1 =  C_{max} - i\delta_1 .  \label{eq:case213}
\end{eqnarray}
Consequently, $[\alpha,\lambda,\mu]$ is an approximate state verifying the two conditions.
%% **************************************************** Subcase 2.2*******************************************
\begin{itemize}
\item Sub-case 2.2:$\sum_{j=1}^{i}p_j - C_{max}^{'\#} > C_{max}^{'\#}$
\end{itemize}
Here, the state 
 $[1-k^{'\#}$, $\max\{L_{max}^{'\#}$,$\sum_{j=1}^{i}p_j - C_{max}^{'\#} +q_i\}$,$\sum_{j=1}^{i}p_j - C_{max}^{'\#}]$ is created by algorithm $A^{'}$at iteration $i$. However, it may be removed when reducing the state subset. Let $[\alpha,\lambda,\mu]$ be the state in $\chi_i^\#$ that is in the same box as $[1-k^{'\#}$, $\max\{L_{max}^{'\#}$,$\sum_{j=1}^{i}p_j - C_{max}^{'\#} +q_i\}$,$\sum_{j=1}^{i}p_j - C_{max}^{'\#}]$. Hence, we have:

\begin{eqnarray}
\lambda & \leq & \max\{ L_{max}^{'\#},\sum_{j=1}^{i}p_j - C_{max}^{'\#} +q_i\} + \delta_2 \nonumber \\
& \leq & \max\{L_{max}^{'}+(i-1)\max\{\delta_1,\delta_2\},\sum_{j=1}^{i}p_j \nonumber\\
&& - (C_{max}^{'} +(i-1)\delta_1)+q_i\} + \delta_2 \nonumber \\
& \leq & \max\{L_{max}^{'},\sum_{j=1}^{i}p_j - C_{max}^{'}+q_i\} \nonumber \\
&&+(i-1)\max\{\delta_1,\delta_2\} + \delta_2 \nonumber \\
& \leq & L_{max} +(i-1)\max\{\delta_1,\delta_2\} + \delta_2 \nonumber \\
& < & L_{max}+ i.\max\{\delta_1,\delta_2\} \label{eq:case221} 
\end{eqnarray}

Moreover, 
\begin{eqnarray}
\mu &\leq & \sum_{j=1}^{i}p_j - C_{max}^{'\#} + \delta_1 \nonumber \\
\end{eqnarray}Since $C_{max}^{'\#} \geq C_{max}^{'} - (i-1) \delta_1$, then the following relation holds
\begin{eqnarray}
\mu &\leq & \sum_{j=1}^{i}p_j - C_{max}^{'}+(i-1)\delta_1 + \delta_1 \leq C_{max}+i\delta_1  \nonumber \\
  \label{eq:case222}
\end{eqnarray}(since $\sum_{j=1}^{i}p_j - C_{max}^{'} \leq C_{max}^{'}$).\\
And,
\begin{eqnarray}
\mu &\geq& \sum_{j=1}^{i}p_j - C_{max}^{'\#} - \delta_1 \geq C_{max}^{'\#} - \delta_1  \nonumber \\ 
&\geq& C_{max} - (i-1)\delta_1 - \delta_1  =  C_{max} - i \delta_1 .  \label{eq:case223}
\end{eqnarray}
Thus, $[\alpha,\lambda,\mu]$ verifies the necessary conditions.

%% ****************************************************Case 3*******************************************

\begin{itemize}
\item \textbf{$3^{rd}$ Case: } $[k,L_{max},C_{max}]= [1-k^{'}$, $\max\{L_{max}^{'}$,$\sum_{j=1}^{i}p_j - C_{max}^{'} +q_i\}$,$\sum_{j=1}^{i}p_j-C_{max}^{'}]$ 
\end{itemize}
Since $[k^{'}$,$L_{max}^{'}$,$C_{max}^{'}]\in \chi_{i-1}$, there exists
$[k^{'\#},L_{max}^{'\#},C_{max}^{'\#} ]\in \chi_{i-1}^\#$, such that:\\
\(L_{max}^{'\#} \leq L_{max}^{'}+(i-1)\max\{\delta_1,\delta_2\}\)
and \(C_{max}^{'} -(i-1) \delta_1 \leq C_{max}^{'\#} \leq C_{max}^{'} +(i-1) \delta_1 .\)\\
Consequently, two sub-cases can occur:\\
%% **************************************************** Subcase 3.1*******************************************
\begin{itemize}
\item Sub-case 3.1:$\sum_{j=1}^{i}p_j - C_{max}^{'\#} \geq  C_{max}^{'\#}$
\end{itemize}
Here, the state 
$[1-k^{'\#}$,$\max\{L_{max}^{'\#}$,$\sum_{j=1}^{i}p_j - C_{max}^{'\#} +q_i\}$,$\sum_{j=1}^{i}p_j - C_{max}^{'\#}]$ is created by algorithm $A^{'}$at iteration $i$. However, it may be removed when reducing the state subset. Let $[\alpha,\lambda,\mu]$ be the state in $\chi_i^\#$ that is in the same box as $[1-k^{'\#}$,$\max\{L_{max}^{'\#}$,$\sum_{j=1}^{i}p_j - C_{max}^{'\#} +q_i\}$,$\sum_{j=1}^{i}p_j - C_{max}^{'\#}]$. Hence, we have:
\begin{eqnarray}
\lambda & \leq & \max\{ L_{max}^{'\#},\sum_{j=1}^{i}p_j - C_{max}^{'\#} +q_i\} + \delta_2 \nonumber \\
& \leq & \max\{L_{max}^{'}+(i-1)max\{\delta_1,\delta_2\},\sum_{j=1}^{i}p_j \nonumber\\
&& - (C_{max}^{'} -(i-1)\delta_1)+q_i\} + \delta_2 \nonumber \\
& \leq & \max\{L_{max}^{'},\sum_{j=1}^{i}p_j \nonumber\\ 
&& - C_{max}^{'}+q_i\}+(i-1)\max\{\delta_1,\delta_2\} + \delta_2 \nonumber \\
& \leq & L_{max} +(i-1)\max\{\delta_1,\delta_2\} + \delta_2 \nonumber \\
& \leq & L_{max}+ i\max\{\delta_1,\delta_2\} \label{eq:case311} 
\end{eqnarray}
and 
\begin{eqnarray}
\mu &\leq & \sum_{j=1}^{i}p_j - C_{max}^{'\#} + \delta_1 \nonumber \\
&&\leq \sum_{j=1}^{i}p_j - (C_{max}^{'} - (i-1)\delta_1) + \delta_1 \nonumber \\
&&\leq C_{max}+i\delta_1 .  \label{eq:case312}
\end{eqnarray}In the other hand, we have
\begin{eqnarray}
\mu &\geq & \sum_{j=1}^{i}p_j - C_{max}^{'\#} - \delta_1 \nonumber\\
&\geq& \sum_{j=1}^{i}p_j - (C_{max}^{'} + (i-1)\delta_1) - \delta_1 \nonumber \\
&\geq& C_{max}-i\delta_1 .  \label{eq:case313}
\end{eqnarray}Thus, $[\alpha,\lambda,\mu]$ fulfills the conditions.
%% **************************************************** Subcase 3.2*******************************************
\begin{itemize}
\item Sub-case 3.2:$\sum_{j=1}^{i}p_j - C_{max}^{'\#} < C_{max}^{'\#}$
\end{itemize}
 Here, the state $[k^{'\#}$, $\max\{L_{max}^{'\#}$,$\sum_{j=1}^{i}p_j - C_{max}^{'\#} +q_i\}$,$C_{max}^{'\#}]$ is created by algorithm $A^{'}$at iteration $i$. However, it may be removed when reducing the state subset. Let $[\alpha,\lambda,\mu]$ be the state in $\chi_i^\#$ that is in the same box as $[k^{'\#}$, $\max\{L_{max}^{'\#}$,$\sum_{j=1}^{i}p_j - C_{max}^{'\#} +q_i\}$,$C_{max}^{'\#}]$. Hence, we have:

\begin{eqnarray}
\lambda & \leq & \max\{ L_{max}^{'\#},\sum_{j=1}^{i}p_j - C_{max}^{'\#} +q_i\} \nonumber \\
&&+ \delta_2 \nonumber \\
& \leq & \max\{L_{max}^{'}+(i-1)\max\{\delta_1,\delta_2\},\sum_{j=1}^{i}p_j \nonumber\\
&&- (C_{max}^{'} -(i-1)\delta_1)+q_i\} + \delta_2 \nonumber \\
& \leq & \max\{L_{max}^{'},\sum_{j=1}^{i}p_j - C_{max}^{'}+q_i\} \nonumber \\
&&+(i-1)\max\{\delta_1,\delta_2\} + \delta_2 \nonumber \\
& \leq & L_{max} +(i-1)\max\{\delta_1,\delta_2\} + \delta_2 \nonumber \\
& \leq & L_{max}+ i\max\{\delta_1,\delta_2\} \label{eq:case321} 
\end{eqnarray}

and
\begin{eqnarray}
\mu &\leq& C_{max}^{'\#} + \delta_1 \leq C_{max}^{'}+(i-1)\delta_1 + \delta_1 \nonumber \\
&\leq& C_{max}+i\delta_1. \nonumber \\
\label{eq:case322}
\end{eqnarray}In the other hand, we have
\begin{eqnarray}
\mu &\geq& C_{max}^{'\#} - \delta_1 \geq \sum_{j=1}^{i}p_j - C_{max}^{'\#} - \delta_1 \nonumber \\ 
&\geq& \sum_{j=1}^{i}p_j - C_{max}^{'} - i\delta_1  =  C_{max} - i\delta_1 .  \label{eq:case323}
\end{eqnarray}
Therefore, $[\alpha,\lambda,\mu]$  fulfills the conditions.\\\\
In conclusion, the statement holds also for level $i$ in the third case, and this completes our inductive proof.
\end{proof}
Based on the lemma, we deduce easily that for every non-dominated state $[k,L_{max},C_{max}]\in \chi_n$, it must remain a close state $[m,L_{max}^\#,C_{max}^\#] \in \chi_n^\#$ such that:
		\[L_{max}^\#≤L_{max}+n.\max\{\delta_1,\delta_2\} ≤ (1+\epsilon).L_{max}, \]
        and
			\[C_{max}^\#≤C_{max}+n.\delta_1 ≤(1+\epsilon).C_{max}.\]
Moreover, it is clear that the FPTAS runs polynomially in $n$ and $1/\epsilon$. The overall complexity of our FPTAS is $O(n^{3}/\epsilon^2)$.

\section{Results}
The following results have been obtained after testing the performance of the proposed algorithms. The code has been done in Java and the experiments were performed on an Intel(R) Core(TM)-i7 with 8GB RAM. We randomly generate five sets of instances, with different numbers of jobs and various processing and delivery times:
\begin{itemize}
\item number of jobs: from 5 to 25, 26 to 50, 51 to 75, 76 to 100 and 100 to 200  
\item processing times : from 1 to 20, 1 to 100 and 1 to 1000
\item delivery times : from 1 to 20, 1 to 100 and 1 to 1000
\end{itemize}
That gave us 135 instances in each set of instances. Finally, the FPTAS has been tested with two values of $\epsilon$: $0.3$ and $0.9$. To ensure the consistency of running times, each test has been run three times.

Figure \ref{fig:quality} presents a comparison of FPTAS and Dynamic Programming. The left part of this figure shows the average size of the Pareto Front ({\em i.e.} the number of solutions) found by the Dynamic Programming algorithm and our FPTAS with the two $\epsilon$ values we used. The sizes are given for our five sets of instances, from small instances (5-25 instances) to bigger ones (100-200 jobs). We can see that the number of solutions decreases as the number of jobs increases. With a lot of jobs, it is more likely to obtain very similar solutions, a lot of them being dominated by others. At the opposite, the number of jobs has no real influence on the Pareto front sizes found by our FPTAS algorithm, whatever the value of $\epsilon$.

On the right part of the same figure are given the average quality of the two objectives of our study: $L_{max}^{A'^{\epsilon }}/L^*_{max}$ and $C_{max}^{A'^{\epsilon }}/C^*_{max}$. We can see that the FPTAS algorithms are finding solutions closer to the optimal ones when the number of jobs is increasing. $C_{max}$ values are closer to the optimal than $L_{max}$ values, which is not a surprise, as $L_{max}$ depends on $C_{max}$. Worth to mention, our FPTAS with $\epsilon=0.3$ gives better results than with $\epsilon=0.9$, which is consistent with the theory.  

\begin{figure*}
   \begin{minipage}[c]{.49\linewidth}
      \includegraphics[height=4cm]{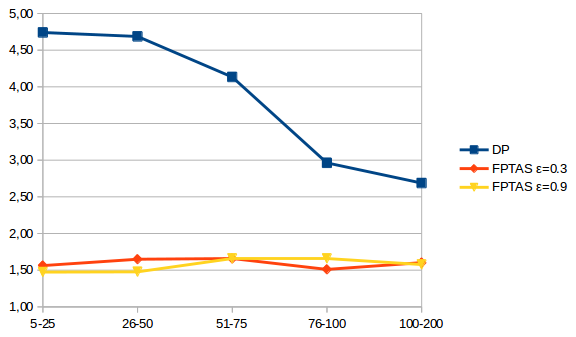}
      \legend{Size of Pareto Front}
   \end{minipage} \hfill
   \begin{minipage}[c]{.49\linewidth}
      \includegraphics[height=4cm]{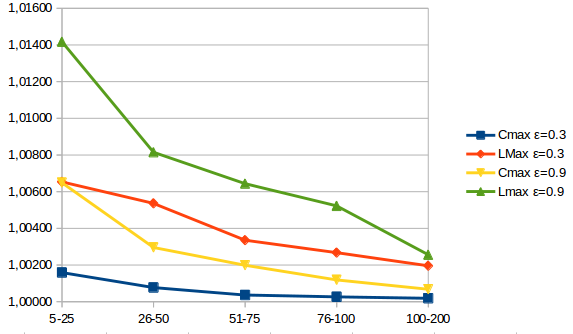}
      \legend{Quality of $C_{max}$ and $L_{max}$}
   \end{minipage} \hfill
   \caption{Quality of our FPTAS algorithm with $\epsilon=0.3$ and $\epsilon=0.9$}
   \label{fig:quality}
\end{figure*}

We have also studied the influence of processing and delivery times, in Table \ref{table:P} and \ref{table:Q}. These tables present average results for our benchmark, considering the 5 sets of instances. Results are presented for $\epsilon=0.3$; the analysis is the same with $\epsilon=0.9$. Table \ref{table:P} shows that instances composed of jobs with various processing times lead to optimal solutions with a bigger Pareto front, as seen in column 2. At the opposite, columns 3 to 5 show that our FPTAS algorithm is not really influenced by this parameter. Table \ref{table:Q} shows that delivery times ranges have more influence on the results : the size of the Pareto front of non-dominated solutions grows faster. Our FPTAS algorithm has the same behavior, and we can see that the results are more close to the optimal with smaller values of $q_i$.

\begin{table}[htbp]
\caption{Quality of FPTAS as a function of processing time ranges (for $\epsilon=0.3$)}
\begin{center}
\begin{tabular}{|l|c|c|c|c|}
\hline
\multirow{2}*{$p_i$ range} & DP size of   & FPTAS size of & \multirow{2}*{$C_{max}^{A'^{\epsilon }}/C^*_{max}$} & \multirow{2}*{$L_{max}^{A'^{\epsilon }}/L^*_{max}$}\\ 
 & Pareto front & Pareto front  & &\\ \hline
1-20  & 2.26 & 1.23 & 1.0006  & 1.003\\ \hline
1-100 & 4.07 & 1.84 & 1.0008 & 1.007\\ \hline
1-500 & 5.65 & 1.73 & 1.0005 & 1.004\\ \hline
\end{tabular}
\label{table:P}
\end{center}
\end{table}

\begin{table}[htbp]
\caption{Quality of FPTAS as a function of delivery time ranges (for $\epsilon=0.3$)}
\begin{center}
\begin{tabular}{|l|c|c|c|c|}
\hline
\multirow{2}*{$q_i$ range} & DP size of   & FPTAS size of & \multirow{2}*{$C_{max}^{A'^{\epsilon }}/C^*_{max}$} & \multirow{2}*{$L_{max}^{A'^{\epsilon }}/L^*_{max}$}\\ 
 & Pareto front & Pareto front & & \\ \hline
1-20  & 2.02 & 1.25 & 1.0007  & 1.001\\ \hline
1-100 & 3.03 & 1.57 & 1.0007 & 1.002\\ \hline
1-500 & 6.93 & 1.99 & 1.0005 & 1.008\\ \hline
\end{tabular}
\label{table:Q}
\end{center}
\end{table}

Computing times are given in Tables \ref{table compj}, \ref{table compp} and \ref{table compq}. They compare our Dynamic Programming algorithm and our FPTAS, considering two $\epsilon$ values : 0.3 and 0.9. All values are in milliseconds. Table \ref{table compj} shows that all algorithms are slower when the number of states is growing. Table \ref{table compp} shows an interesting result: while the Dynamic Programming algorithm becomes slower when the processing times ranges are growing, the FPTAS has an opposite behavior. The FPTAS with $\epsilon=0.3$ is even slower than the exact algorithm for the smallest range of processing times. 
Table \ref{table compq} shows that the delivery times ranges have a smaller influence on the Dynamic Programming algorithm: computing times are growing, but slower than the delivery times ranges. The FPTAS computing times are also growing in function of the delivery times ranges. Note that tables \ref{table compp} and \ref{table compq} are based on mean values from all the set of instances.

\begin{table}[htbp]
\caption{Average computing times vs size of instances (ms)}
\begin{center}
\begin{tabular}{|l|c|c|c|}
\hline
\multirow{2}*{\#jobs} & \multirow{2}*{DP} & \multicolumn{2}{c|}{FPTAS}\\ \cline{3-4}
& & $\epsilon=0.3$&$\epsilon=0.9,$ \\ \hline
 5-25  & 67 & 0.9 &0.3  \\ \hline
26-50 & 1278 & 5.4 & 0.9  \\ \hline
51-75 & 7917 & 22 & 3.4  \\ \hline
76-100 & 24937 & 58 & 7.8 \\ \hline
100-200 & 164332 & 281 & 32 \\ \hline
\hline
\end{tabular}
\label{table compj}
\end{center}
\end{table}

\begin{table}[htbp]
\caption{Average computing times (ms) vs processing time ranges}
\begin{center}
\begin{tabular}{|l|c|c|c|}
\hline
\multirow{2}*{$p_i ranges$} & \multirow{2}*{DP} & \multicolumn{2}{c|}{FPTAS}\\ \cline{3-4}
& & $\epsilon=0.3$&$\epsilon=0.9,$ \\ \hline
 1-20  & 91 & 144 & 15 \\ \hline
1-100 & 2668 & 51 & 7  \\ \hline
1-500 & 116362 & 25 & 5 \\ \hline
\hline
\end{tabular}
\label{table compp}
\end{center}
\end{table}

\begin{table}[htbp]
\caption{Average computing times (ms) vs delivery time ranges}
\begin{center}
\begin{tabular}{|l|c|c|c|}
\hline
\multirow{2}*{$q_i ranges$} & \multirow{2}*{DP} & \multicolumn{2}{c|}{FPTAS}\\ \cline{3-4}
& & $\epsilon=0.3$&$\epsilon=0.9,$ \\ \hline
 1-20  & 31447 & 26 & 5 \\ \hline
1-100 & 40482 & 46 & 7  \\ \hline
1-500 & 47190 & 149 & 15 \\ \hline
\hline
\end{tabular}
\label{table compq}
\end{center}
\end{table}

\section{Conclusions and Perspectives}

The two-parallel machines scheduling problem has been considered to minimize the maximum lateness and the makespan. We have proposed an exact algorithm (based on dynamic algorithm) to generate the complete Pareto Frontier in a pseudo-polynomial time. Then, we present an FPTAS (Fully Polynomial Time Approximation Scheme) to generate an approximate Pareto Frontier, based on the conversion of the exact dynamic programming. For the proposed algorithms, we randomly generated several instances with different ranges, and, for each job $J_i$, its processing time $p_i$ and delivery time $q_i$ are sets to be integer numbers.The results of the experiments showed that the proposed algorithms for the considered problem are very efficient. It is clear that optimizing the maximum lateness $(L_{max})$ implies to minimize implicitly the makespan $(C_{max})$. Moreover, the values of $\epsilon$ and processing and delivery times play an important role in the results (i.e., big processing times, small delivery times and big $\epsilon$ make the FPTAS faster and vice versa).\\
In our future works, the study of the multiple-machine scheduling problems seems to be a challenging perspective in the extension of our work.

%
% ---- Bibliography ----
%
\bibliographystyle{splncs}

\end{document}